\documentclass[11pt,onecolumn]{IEEEtran}
\usepackage[utf8]{inputenc}
\usepackage{mathtools}
\usepackage{amsmath}
\usepackage{amsthm}
\usepackage{amssymb}
\usepackage{geometry}
\usepackage[bookmarks=false,
 breaklinks=false,pdfborder={0 0 1},backref=false,colorlinks=false]
 {hyperref}

\makeatletter
\@ifundefined{date}{}{\date{}}

\newcommand{\E}{\mathbb E}

\newcommand{\cC}{\mathcal C}
\newcommand{\cK}{\mathcal K}

\newcommand{\cT}{\mathcal T}

\newcommand{\X}{\mathcal X}
\newcommand{\Y}{\mathcal Y}

\theoremstyle{plain}
\newtheorem{thm}{\protect\theoremname}\theoremstyle{definition}
\newtheorem{defn}{\protect\definitionname}\theoremstyle{plain}
\theoremstyle{remark}
\newtheorem{rem}{\protect\remarkname}\theoremstyle{plain}
\newtheorem{prop}{\protect\propositionname}
\newtheorem{lem}{\protect\lemmaname}
\providecommand{\conjecturename}{Conjecture}
\providecommand{\definitionname}{Definition}
\providecommand{\lemmaname}{Lemma}
\providecommand{\propositionname}{Proposition}
\providecommand{\remarkname}{Remark}
\providecommand{\theoremname}{Theorem}

\makeatother

\providecommand{\definitionname}{Definition}
\providecommand{\lemmaname}{Lemma}
\providecommand{\propositionname}{Proposition}
\providecommand{\remarkname}{Remark}
\providecommand{\theoremname}{Theorem}

\begin{document}
\title{Exact Common Information and Exact Channel Synthesis for Correlated
Discrete Sources}
\author{Lei Yu\thanks{L. Yu is with the School of Statistics and Data Science, LPMC, KLMDASR,
and LEBPS, Nankai University, Tianjin 300071, China (e-mail: leiyu@nankai.edu.cn).
This work was supported by the National Key Research and Development
Program of China under grant 2023YFA1009604 and the NSFC under grant
62101286.} }
\maketitle
\begin{abstract}
In this paper, we provide a single-letter characterization of the
exact common information for a pair of correlated discrete sources,
as well as a single-letter characterization of the admissible region
for the shared randomness rate and the communication rate in exact
channel synthesis. 
\end{abstract}


\section{Introduction}

\subsection{Common information }

The common information problem asks for the minimum amount of common
randomness needed to generate two correlated random variables at two
separate terminals; see Fig. \ref{fig:dss}. This minimum amount of
common randomness is coined as the common information between the
two correlated random variables. In Wyner's formulation \cite{WynerCI},
Wyner considered approximate generation in relative entropy and obtained
the common information for a bivariate joint distribution $\pi_{XY}$,
\[
C_{{\rm Wyner}}(\pi_{XY})=\inf_{W:\,P_{XY}=\pi_{XY},\,X-W-Y}I(X,Y;W).
\]

The exact common information problem, introduced by Kumar, Li, and
El Gamal \cite{KLE2014}, is different, where the synthesized distribution
must equal the target distribution exactly, rather than merely approach
it asymptotically in a weak metric. The common random variable is
first generated and then independently processed at the two terminals.
The exact common information is the minimum asymptotic rate of this
common randomness. We next provide a mathematical formulation of this
problem.

Let $\pi_{XY}$ be a probability distribution on a standard Borel
product space. At blocklength $n$, an exact synthesis code consists
of a discrete random variable $W$ and conditional distributions $P_{X^{n}|W},P_{Y^{n}|W},$
such that 
\[
P_{X^{n}Y^{n}}=\sum_{w}P_{W}(w)P_{X^{n}|W=w}P_{Y^{n}|W=w}=\pi^{\otimes n}_{XY},
\]
where the conditional independence relation $X^{n}-W-Y^{n}$ holds.

The exact common information is the asymptotic normalized entropy
of the smallest such common random variable. The variable-length formulation
is equivalent because for a prefix-free code, the expected number
$L(W)$ of output bits satisfies $H(W)\leq L(W)<H(W)+1,$ and therefore
the difference is negligible after normalization by $n$. Accordingly,
\begin{equation}
C_{{\rm Exact}}(\pi_{XY})=\lim_{n\to\infty}\frac{1}{n}\inf\left\{ H(W):P_{X^{n}Y^{n}}=\pi^{\otimes n}_{XY},\ X^{n}-W-Y^{n}\right\} .\label{eq:exact-multiletter}
\end{equation}
See detailed explanation in \cite{KLE2014}.

A natural interpretation of exact common information can be described
via the following thought experiment. Consider a random particle (or
planet) that, at a certain instant, decomposes into multiple components.
These components inherit shared common randomness and then evolve
independently, governed jointly by this common randomness as well
as their respective individual randomness. After some time has elapsed,
given the observed joint distribution of the components, one may wish
to estimate the amount of common randomness they possess. Or conversely,
given the quantity of common randomness shared among them, we aim
to characterize the set of feasible joint distributions these components
can attain. This is exactly the common information problem.

The exact common information is no smaller than Wyner's common information.
However, unlike Wyner's common information, the single-letter characterization
of the exact common information is rarely known, except for a certain
class of sources satisfying $C_{\mathrm{Exact}}=C_{{\rm Wyner}}$
\cite{KLE2014,Vellambi2018,YuTan2020_exact} and the doubly symmetric
binary source (DSBS) \cite{YuTan2020_exact}. For continuous log-concave
sources, Li and El Gamal \cite{LiElgamal2017} showed that the exact
common information is finite and an explicit upper bound was given.
Yu and Tan \cite{YuTan2018,yu2020corrections,YuTan2020_exact} introduced
the notion of Rényi common information, which is defined as the minimum
common rate when the the relative entropy is replaced by more general
divergences---the family of Rényi divergences. The family of Rényi
common information unifies Wyner's common information and exact common
information since it includes them as two special cases with the Rényi
order equal to $1$ and $\infty$, respectively. Li and El Gamal's
upper bound in \cite{LiElgamal2017} play a key role in proving this
equivalence for continuous sources. By this equivalence and the Rényi
soft-covering argument, Yu and Tan \cite{YuTan2020_exact} significantly
improved Li and El Gamal's upper bound for correlated Gaussian sources
and conjectured this new bound is tight. This conjecture was confirmed
in positive by the present author recently \cite{yu2026exact}, yielding
that for a bivariate Gaussian distribution $\pi_{\rho}$, 
\[
C_{{\rm Exact}}(\pi_{\rho})=\frac{1}{2}\ln\frac{1+\rho}{1-\rho}+\frac{\rho}{1+\rho}.
\]

\begin{figure}
\centering \setlength{\unitlength}{0.05cm} 
{ \begin{picture}(100,60) 
\put(5,30){\line(1,0){20}} \put(25,30){\vector(1,1){18}}
\put(25,30){\vector(1,-1){18}} \put(44,42){\framebox(30,12){$P_{X^{n}|W}$}}
\put(44,6){\framebox(30,12){$P_{Y^{n}|W}$}} \put(74,48){\vector(1,0){22}}
\put(74,12){\vector(1,0){22}} \put(10,33){%
\mbox{%
$W$%
}} \put(80,50){%
\mbox{%
$X^{n}$%
}} \put(80,14){%
\mbox{%
$Y^{n}$%
}} \end{picture}}

\caption{ The common information problem, also known as the distributed source
synthesis problem.}\label{fig:dss}
\end{figure}
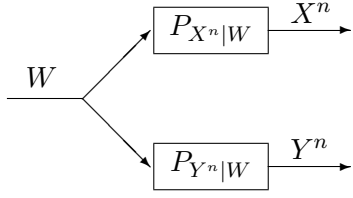

\subsection{Distributed channel synthesis }

Common information has natural applications in distributed channel
synthesis \cite{Bennett02,Win02,cuff13,bennett14quantum,Harsha10}.
The latter problem, illustrated in Fig. \ref{fig:dcs}, refers to
the problem of determining the minimum communication rate required
to generate a bivariate source $\left\{ \left(X^{n},Y^{n}\right)\right\} _{n\in\mathbb{N}}$
with $X^{n}$ generated at the encoder and $Y^{n}$ generated at the
decoder such that the induced joint distribution $P_{X^{n}Y^{n}}$
approximately or exactly equals $\pi^{\otimes n}_{XY}$ for all $n\in\mathbb{N}$.
In this paper, we focus on exact synthesis, equivalently, requiring
$P_{X^{n}Y^{n}}=\pi^{\otimes n}_{XY}$. When there is no shared randomness,
the exact channel synthesis problem reduces to the exact common information
problem.

Consider the distributed source simulation setup depicted in Fig.
\ref{fig:dcs}. A sender and a receiver share a uniformly distributed
source of randomness\footnote{For simplicity, we assume that $e^{nR_{0}}$ is integers.}
$K\sim\mathrm{Unif}(\cK),\cK:=[e^{nR_{0}}]$. The sender has access
to a memoryless source $X^{n}\sim\pi^{\otimes n}_{X}$ that is independent
of $K$, and wants to transmit information about the correlation between
correlated sources $\left(X^{n},Y^{n}\right)\sim\pi^{\otimes n}_{XY}$
to the receiver. Given the shared randomness and the correlation information
from the sender, the receiver generates a memoryless source $Y^{n}\sim\pi^{\otimes n}_{Y|X}(\cdot|X^{n})$.
Specifically, given $X^{n}$ and $K$, the sender generates a ``message''
(i.e., a discrete random variable) $M$ by a random mapping $P_{M|X^{n}K}$,
and then sends it to the receiver error free. Upon accessing to $K$
and receiving $M$, the receiver generates a source $Y^{n}$ by a
random mapping $P_{Y^{n}|MK}$. The joint distribution induced by
this code is 
\begin{align*}
 & P_{X^{n}KMY^{n}}:=P_{X^{n}}P_{K}P_{M|X^{n}K}P_{Y^{n}|MK}.
\end{align*}
Now we would like to determine the minimum amount of communication
such that $P_{X^{n}Y^{n}}=\pi^{\otimes n}_{XY}$ (or equivalently,
$P_{Y^{n}|X^{n}}=\pi^{\otimes n}_{Y|X}$).
\begin{defn}
The admissible region of shared randomness rate and communication
rate for the exact channel synthesis problem is defined as 
\begin{align}
 & \mathcal{R}_{\mathrm{Exact}}(\pi_{XY})\nonumber \\
 & =\mathrm{cl}\bigcup_{n\ge1}\left\{ \begin{array}{l}
(R_{0},R):\exists(P_{M|X^{n}K},P_{Y^{n}|MK})\textrm{ s.t.}\\
\qquad P_{Y^{n}|X^{n}}=\pi^{\otimes n}_{Y|X},\\
\qquad R\ge\frac{1}{n}H(M|K)
\end{array}\right\} .\label{eq:-6}
\end{align}
\end{defn}
The original operational definition of this admissible region adopts
a variable-length coding formulation, where the message $M$ is compressed
via a prefix-free code, and the communication rate $R$ is quantified
as the bit rate of this code. However, same as the common information
setting, this variable-length formulation is equivalent to the formulation
above; see Proposition 1 of \cite{YuTan2020b}.

In contrast to exact channel synthesis, total variation (TV)-approximate
synthesis only requires the TV distance between the empirical distribution
$P_{X^{n}Y^{n}}$ and the target distribution $\pi^{\otimes n}_{XY}$
to vanish asymptotically. Early studies by Bennett et al. \cite{Bennett02}
and Winter \cite{Win02} investigated exact and TV-approximate channel
synthesis under the assumption of unlimited shared encoder--decoder
randomness, proving that the minimal asymptotic communication rate
for both synthesis schemes equals the mutual information $I(X;Y)$
of the target distribution $(X,Y)\sim\pi_{XY}$. For TV-approximate
synthesis, Cuff \cite{cuff13} and Bennett et al. \cite{bennett14quantum}
completely characterized the fundamental tradeoff between communication
rate and shared randomness rate. 

As for exact synthesis, Harsha et al. \cite{Harsha10} adopted a
rejection sampling framework to analyze one-shot exact synthesis for
discrete sources, establishing a bounded shared randomness cost with
a mild increment in expected description length. Li and El Gamal \cite{LiElgamal2018}
further refined the finite shared randomness upper bound for exact
synthesis in the discrete setting with a negligible communication
rate penalty. Sriramu and Wagner \cite{sriramu2024optimal} characterized
the second-order rate of the communication rate for exact synthesis,
when the shared randomness rate is assumed to be infinite. In addition,
exact channel synthesis with fixed-length coding under the assumption
of zero or infinite shared randomness was studied by Cubitt, Leung,
Matthews, and Winter \cite{Cubitt11}. They provided a fully characterization
of the communicate rate for correlated sources defined on finite alphabets,
but the problem formulation is not applied to continuous sources,
since from their characterization we can see that fixed-length coding
in general cannot synthesize a continuous channel with continuous
input at a finite shared randomness  rate, even under the assumption
of infinite shared randomness. 

Prior works mainly focused on exact synthesis under extreme randomness
conditions, except for \cite{YuTan2020b}. Yu and Tan \cite{YuTan2020b}
fully characterized the optimal tradeoff between shared randomness
rate and communication rate for exact synthesis of the DSBS (i.e.,
the binary symmetric channel with the binary uniform input), and verified
that exact synthesis requires a strictly higher communication rate
than its TV-approximate counterpart. This is the first example for
which the optimal rate tradeoff is known to be different from the
TV-approximate counterpart. The Gaussian setting remained an open
problem until recent work by the author \cite{yu2026exact}. 

Other related works include \cite{cao2024channelsimulationfinite,li2026largedevrst,oufkir2026exponentsrandomness,Oufkir2026}.
We refer readers to the comprehensive reviews on common information
and channel synthesis presented in \cite{yu2022common} and \cite{li2024channelsimulation}.

\begin{figure*}
\centering \setlength{\unitlength}{0.06cm} { \begin{picture}(140,35)
\put(-5,10){\vector(1,0){30}} \put(-10,13){%
\mbox{%
$X^{n}\sim\pi^{\otimes n}_{X}$%
}} \put(25,4){\framebox(30,12){$P_{M|X^{n}K}$}} \put(55,10){\vector(1,0){30}}
\put(65,13){%
\mbox{%
$M$%
}} \put(85,4){\framebox(30,12){$P_{Y^{n}|MK}$}} \put(115,10){\vector(1,0){20}}
\put(120,13){%
\mbox{%
$Y^{n}\sim\pi^{\otimes n}_{Y|X}(\cdot|X^{n})$%
}} \put(40,27){\vector(0,-1){11}} \put(-5,30){%
\mbox{%
$K\sim\mathrm{Unif}[e^{nR_{0}}]$%
}} \put(100,27){\vector(0,-1){11}} \put(-5,27){\line(1,0){105}}
\end{picture}}

\caption{ The exact channel synthesis problem. }\label{fig:dcs}
\end{figure*}
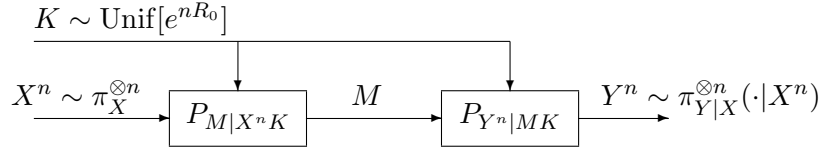

\subsection{Multi-letter characterizations}

All alphabets in this paper are finite, and all logarithms are natural.
For probability distributions $P$ on $\mathcal{X}$ and $Q$ on $\mathcal{Y}$,
let $\cC(P,Q)$ denote the set of their couplings. For a nonnegative,
possibly extended-valued cost $c:\mathcal{X}\times\mathcal{Y}\to[0,\infty]$,
define 
\begin{align}
\cT_{c}(P,Q) & :=\sup_{R\in\cC(P,Q)}\E_{R}[c(X,Y)],\label{eq:transport}\\
\Phi_{c}(P,Q) & :=\cT_{c}(P,Q)-H(P)-H(Q).\label{eq:phi}
\end{align}
We use the conventions $0\log0=0$ and $0\cdot\infty=0$. For finite-valued
costs, the supremum in \eqref{eq:transport} is attained.

Given a source distribution $\pi=\pi_{XY}$ on $\mathcal{X}\times\mathcal{Y}$,
write 
\[
c_{\pi}(x,y):=\log\frac{1}{\pi(x,y)},
\]
where $c_{\pi}(x,y)=+\infty$ when $\pi(x,y)=0$. In this case, $\cT_{c_{\pi}}(P,Q)$
is known as the maximum cross-entropy over all couplings of $(P,Q)$.
Define a functional $\Gamma$ as 
\begin{equation}
\Gamma(\pi):=\inf_{\substack{P_{WXY}=P_{W}P_{X|W}P_{Y|W},\\
P_{XY}=\pi
}
}\E_{W}\!\left[\Phi_{c_{\pi}}(P_{X|W},P_{Y|W})\right],\label{eq:gamma}
\end{equation}
where $W$ ranges over finite alphabets. Equivalently, 
\begin{equation}
\Gamma(\pi)=\inf_{\substack{X-W-Y,\\
P_{XY}=\pi
}
}\left\{ -H(XY|W)+\sum_{w}P_{W}(w)\cT_{c_{\pi}}(P_{X|w},P_{Y|w})\right\} .\label{eq:gamma-equivalent}
\end{equation}
Here the entropy is evaluated under the synthesis law $P_{W}P_{X|W}P_{Y|W}$,
and hence, $H(XY|W)=H(X|W)+H(Y|W)$. This functional was first introduced
in \cite[Theorem~1]{YuTan2020_exact} to give a multi-letter characterization
of the exact common information. 
\begin{prop}[{\cite[Theorem~1]{YuTan2020_exact}}]
\label{prop:equivalence} For a source with distribution $\pi=\pi_{XY}$
defined on a finite alphabet, 
\begin{equation}
C_{\mathrm{Exact}}(\pi)=\lim_{n\to\infty}\frac{1}{n}\Gamma(\pi^{\otimes n}).\label{eq:equivalence}
\end{equation}
\end{prop}
To derive single-letter characterization, it suffices to prove additivity
of $\Gamma(\pi^{\otimes n})$. 

We next provide a multi-letter characterization for exact channel
synthesis. Define a rate region: 
\begin{align}
\mathcal{R}(\pi) & :=\left\{ \begin{array}{rcl}
(R_{0},R)\in[0,\infty)^{2} & : & \exists P_{W}P_{X|W}P_{Y|W}\textrm{ s.t. }\\
P_{XY} & = & \pi,\\
R & \ge & I(W;X),\\
R+R_{0} & \ge & \E_{W}\!\left[\Phi_{c_{\pi}}(P_{X|W},P_{Y|W})\right]
\end{array}\right\} .\label{eq:UB-1}
\end{align}
Then, 
\begin{align}
\frac{1}{n}\mathcal{R}(\pi^{\otimes n}) & =\left\{ \begin{array}{rcl}
(R_{0},R)\in[0,\infty)^{2} & : & \exists P_{W}P_{X^{n}|W}P_{Y^{n}|W}\textrm{ s.t. }\\
P_{X^{n}Y^{n}} & = & \pi^{\otimes n},\\
R & \ge & \frac{1}{n}I(W;X^{n}),\\
R+R_{0} & \ge & \frac{1}{n}\E_{W}\!\left[\Phi_{c_{\pi^{\otimes n}}}(P_{X^{n}|W},P_{Y^{n}|W})\right]
\end{array}\right\} .\label{eq:UB-1-1}
\end{align}

\begin{prop}[{\cite[Theorem~1]{YuTan2020b}}]
\label{prop:equivalence-1} For a source with distribution $\pi=\pi_{XY}$
defined on a finite alphabet, 
\[
\mathcal{R}_{\mathrm{Exact}}(\pi)=\mathrm{cl}\bigcup_{n\ge1}\frac{1}{n}\mathcal{R}(\pi^{\otimes n}).
\]
\end{prop}

\subsection{Main results}

Our first main contribution is the additivity of $\Gamma$. 
\begin{thm}
\label{thm:tensorization} For arbitrary sources $\pi_{1}$ and $\pi_{2}$
defined on the finite-alphabet $\mathcal{X}\times\mathcal{Y}$, 
\begin{equation}
\Gamma(\pi_{1}\otimes\pi_{2})=\Gamma(\pi_{1})+\Gamma(\pi_{2}).\label{eq:tensorization}
\end{equation}
Consequently, for every positive integer $n$, 
\begin{equation}
\Gamma(\pi^{\otimes n})=n\Gamma(\pi).\label{eq:n-tensorization}
\end{equation}
\end{thm}
Combining this additivity property with Yu--Tan's multi-letter characterization
in Proposition \ref{prop:equivalence} yields a single-letter characterization
of the exact common information. 
\begin{thm}
\label{thm:For-a-source}For a source with distribution $\pi=\pi_{XY}$
defined on a finite alphabet, 
\begin{equation}
C_{\mathrm{Exact}}(\pi)=\Gamma(\pi).\label{eq:ECI}
\end{equation}
\end{thm}
\begin{rem}
Yu and Tan \cite{YuTan2020_exact} showed that the $\infty$-Rényi
common information coincides with the exact common information, and
thus, the $\infty$-Rényi common information also admits the same
expression in \eqref{eq:ECI}. 
\end{rem}

Obverse that 
\begin{equation}
\mathcal{R}(\pi)=\{(R_{0},R)\in[0,\infty)^{2}:(R_{0}+R,R)\in\mathcal{R}_{\mathrm{S}}(\pi)\},\label{eq:sum-rate}
\end{equation}
where 
\begin{align}
\mathcal{R}_{\mathrm{S}}(\pi) & :=\left\{ \begin{array}{rcl}
(R_{\mathrm{S}},R)\in[0,\infty)^{2} & : & \exists P_{W}P_{X|W}P_{Y|W}\textrm{ s.t. }\\
P_{XY} & = & \pi,\\
R & \ge & I(W;X),\\
R_{\mathrm{S}} & \ge & \E_{W}\!\left[\Phi_{c_{\pi}}(P_{X|W},P_{Y|W})\right]
\end{array}\right\} .\label{eq:UB-1-2}
\end{align}
The proof idea of Theorem \ref{thm:tensorization} can be applied
to exact channel synthesis, leading to our second main contribution. 
\begin{thm}
\label{thm:ECS} For arbitrary sources $\pi_{1}$ and $\pi_{2}$ defined
on the finite-alphabet $\mathcal{X}\times\mathcal{Y}$, 
\begin{equation}
\mathcal{R}_{\mathrm{S}}(\pi_{1}\otimes\pi_{2})=\mathcal{R}_{\mathrm{S}}(\pi_{1})+\mathcal{R}_{\mathrm{S}}(\pi_{2}).\label{eq:tensorization-1}
\end{equation}
Consequently, for every positive integer $n$, 
\begin{equation}
\mathcal{R}_{\mathrm{S}}(\pi^{\otimes n})=n\mathcal{R}_{\mathrm{S}}(\pi).\label{eq:n-tensorization-1}
\end{equation}
\end{thm}
Combining this additivity property with \eqref{eq:sum-rate} and Yu--Tan's
multi-letter characterization in Proposition \ref{prop:equivalence-1}
yields a single-letter characterization of the admissible rate region
for exact channel synthesis. 
\begin{thm}
\label{thm:For-a-source-1}For a source with distribution $\pi=\pi_{XY}$
defined on a finite alphabet, 
\begin{equation}
\mathcal{R}_{\mathrm{Exact}}(\pi)=\mathcal{R}(\pi).\label{eq:Gaussian}
\end{equation}
\end{thm}

In fact, by examining the proofs, Theorems \ref{thm:tensorization}
and \ref{thm:ECS} and the converse parts of Theorem \ref{thm:For-a-source}
and \ref{thm:For-a-source-1} can be extended to Polish spaces $\mathcal{X},\mathcal{Y}$
in the present form, where the auxiliary random variable $W$ is defined
on an arbitrary Polish space. However, the achievability parts of
Theorem \ref{thm:For-a-source} and \ref{thm:For-a-source-1} was
proven by constructing random codes for the $\infty$-Rényi divergence
setting and establishing the equivalence between the exact information
quantities and the $\infty$-Rényi counterparts, which were extended
to countablely infinite sources and continuous sources in \cite{YuTan2020_exact}
and \cite{YuTan2020b} under certain additional assumptions. 

\section{Proof of Theorem \ref{thm:tensorization}}

In this section, we prove Theorem \ref{thm:tensorization} by using
an opposite-order coupling, which differs from the same-order coupling
constructed in \cite{YuTan2020_exact}.

\subsection{Opposite-order couplings}
\begin{lem}[Opposite-order gluing]
\label{lem:gluing} Let $P_{X_{1}X_{2}}$ and $Q_{Y_{1}Y_{2}}$ be
arbitrary distributions. For every $y_{2}$, let $R(\cdot|y_{2})\in\cC\bigl(P_{X_{1}},Q_{Y_{1}|Y_{2}=y_{2}}\bigr),$
and for every $x_{1}$, let $S(\cdot|x_{1})\in\cC\bigl(P_{X_{2}|X_{1}=x_{1}},Q_{Y_{2}}\bigr).$
Define 
\begin{equation}
T(x_{1},x_{2},y_{1},y_{2}):=R(x_{1},y_{1}|y_{2})S(x_{2},y_{2}|x_{1}).\label{eq:gluing-plan}
\end{equation}
Then, $T\in\cC\bigl(P_{X_{1}X_{2}},Q_{Y_{1}Y_{2}}\bigr).$ Moreover,
$T_{X_{1}Y_{2}}=P_{X_{1}}\otimes Q_{Y_{2}}.$ 
\end{lem}
\begin{rem}
On the general Polish space $\X_{i},\Y_{i},i=1,2,...,n$, if $P_{X_{i}Y_{i}|X^{i-1}Y^{n}_{i+1}}\in\cC\bigl(P_{X_{i}|X^{i-1}},P_{Y_{i}|Y^{n}_{i+1}}\bigr),$
then
\begin{equation}
\prod^{n}_{i=1}P_{X_{i}Y_{i}|X^{i-1}Y^{n}_{i+1}}\in\cC\bigl(P_{X^{n}},P_{Y^{n}}\bigr),\label{eq:coupling}
\end{equation}
where $P_{X_{i}Y_{i}|X^{i-1}Y^{n}_{i+1}}$ can be decomposed as 
\[
P(\mathrm{d}x_{i},\mathrm{d}y_{i}|x^{i-1},y^{n}_{i+1})=P(\mathrm{d}x_{i}|x^{i-1})P(\mathrm{d}y_{i}|x^{i},y^{n}_{i+1})
\]
and thus, the product in \eqref{eq:coupling} is understood as 
\[
\prod^{n}_{i=1}P(\mathrm{d}x_{i}|x^{i-1})P(\mathrm{d}y_{i}|x^{i},y^{n}_{i+1})=P(\mathrm{d}x^{n})\prod^{n}_{i=1}P(\mathrm{d}y_{i}|x^{i},y^{n}_{i+1}).
\]
Moreover, for each $i$, $X^{i}$ and $Y^{n}_{i+1}$ are independent
under this coupling. 
\end{rem}
\begin{IEEEproof}
Observe that 
\begin{align*}
\sum_{y_{1},y_{2}}T(x_{1},x_{2},y_{1},y_{2}) & =\sum_{y_{2}}\left(\sum_{y_{1}}R(x_{1},y_{1}|y_{2})\right)S(x_{2},y_{2}|x_{1})\\
 & =P_{X_{1}}(x_{1})\sum_{y_{2}}S(x_{2},y_{2}|x_{1})\\
 & =P_{X_{1}}(x_{1})P_{X_{2}|X_{1}}(x_{2}|x_{1})\\
 & =P_{X_{1}X_{2}}(x_{1},x_{2}).
\end{align*}
Similarly, its $(Y_{1},Y_{2})$-marginal satisfies 
\begin{align*}
\sum_{x_{1},x_{2}}T(x_{1},x_{2},y_{1},y_{2}) & =Q_{Y_{1}Y_{2}}(y_{1},y_{2}).
\end{align*}
Hence, $T\in\cC\bigl(P_{X_{1}X_{2}},Q_{Y_{1}Y_{2}}\bigr).$ 

Observe that 
\begin{align*}
\sum_{x_{2},y_{1}}T(x_{1},x_{2},y_{1},y_{2}) & =\sum_{x_{2}}\left(\sum_{y_{1}}R(x_{1},y_{1}|y_{2})\right)S(x_{2},y_{2}|x_{1})\\
 & =P_{X_{1}}(x_{1})\sum_{x_{2}}S(x_{2},y_{2}|x_{1})\\
 & =P_{X_{1}}(x_{1})Q_{Y_{2}}(y_{2}).
\end{align*}
Similarly, its $(X_{1},Y_{2})$ marginal satisfies $T_{X_{2}Y_{1}}=P_{X_{2}}\otimes Q_{Y_{1}}$.
\end{IEEEproof}

\subsection{Superadditivity of $\Phi$}
\begin{lem}[Superadditivity of $\Phi$]
\label{lem:superadditivity} Let $P=P_{X_{1}X_{2}}$ and $Q=Q_{Y_{1}Y_{2}}$
be arbitrary distributions. For costs $c_{i}:\mathcal{X}_{i}\times\mathcal{Y}_{i}\to[0,\infty]$,
$i=1,2$, set 
\[
c_{12}(x_{1},x_{2},y_{1},y_{2}):=c_{1}(x_{1},y_{1})+c_{2}(x_{2},y_{2}).
\]
Then 
\begin{align}
\cT_{c_{12}}(P,Q)\geq & \sum_{y_{2}}Q_{Y_{2}}(y_{2})\cT_{c_{1}}\bigl(P_{X_{1}},Q_{Y_{1}|Y_{2}=y_{2}}\bigr)\nonumber \\
 & +\sum_{x_{1}}P_{X_{1}}(x_{1})\cT_{c_{2}}\bigl(P_{X_{2}|X_{1}=x_{1}},Q_{Y_{2}}\bigr).\label{eq:transport-1}
\end{align}
Moreover, 
\begin{align}
\Phi_{c_{12}}(P,Q)\geq & \sum_{y_{2}}Q_{Y_{2}}(y_{2})\Phi_{c_{1}}\bigl(P_{X_{1}},Q_{Y_{1}|Y_{2}=y_{2}}\bigr)\nonumber \\
 & +\sum_{x_{1}}P_{X_{1}}(x_{1})\Phi_{c_{2}}\bigl(P_{X_{2}|X_{1}=x_{1}},Q_{Y_{2}}\bigr).\label{eq:phi-1}
\end{align}
\end{lem}
\begin{IEEEproof}
First assume that both costs are finite. For every $y_{2}$, let $R(\cdot|y_{2})\in\cC\bigl(P_{X_{1}},Q_{Y_{1}|Y_{2}=y_{2}}\bigr)$
attain $\cT_{c_{1}}\bigl(P_{X_{1}},Q_{Y_{1}|Y_{2}=y_{2}}\bigr)$,
and for every $x_{1}$, let $S(\cdot|x_{1})\in\cC\bigl(P_{X_{2}|X_{1}=x_{1}},Q_{Y_{2}}\bigr)$
attain $\cT_{c_{2}}\bigl(P_{X_{2}|X_{1}=x_{1}},Q_{Y_{2}}\bigr)$.
Then, by Lemma \ref{lem:gluing}, the distribution $T$ constructed
in \eqref{eq:gluing-plan} forms a coupling of $(P,Q)$. 

The first-coordinate cost induced by $T$ is 
\begin{align*}
\E_{T}[c_{1}(X_{1},Y_{1})] & =\sum_{y_{2}}Q_{Y_{2}}(y_{2})\sum_{x_{1},y_{1}}R(x_{1},y_{1}|y_{2})c_{1}(x_{1},y_{1})\\
 & =\sum_{y_{2}}Q_{Y_{2}}(y_{2})\cT_{c_{1}}\bigl(P_{X_{1}},Q_{Y_{1}|Y_{2}=y_{2}}\bigr).
\end{align*}
Similarly, the second-coordinate cost is 
\begin{align*}
\E_{T}[c_{2}(X_{2},Y_{2})] & =\sum_{x_{1}}P_{X_{1}}(x_{1})\sum_{x_{2},y_{2}}S(x_{2},y_{2}|x_{1})c_{2}(x_{2},y_{2})\\
 & =\sum_{x_{1}}P_{X_{1}}(x_{1})\cT_{c_{2}}\bigl(P_{X_{2}|X_{1}=x_{1}},Q_{Y_{2}}\bigr).
\end{align*}
Adding these identities and using $T\in\cC(P,Q)$ yields \eqref{eq:transport-1}.

We next apply the entropy chain rule in opposite directions: 
\begin{align*}
H(P) & =H(P_{X_{1}})+\sum_{x_{1}}P_{X_{1}}(x_{1})H(P_{X_{2}|X_{1}=x_{1}}),\\
H(Q) & =H(Q_{Y_{2}})+\sum_{y_{2}}Q_{Y_{2}}(y_{2})H(Q_{Y_{1}|Y_{2}=y_{2}}).
\end{align*}
Subtracting these from \eqref{eq:transport-1} yields \eqref{eq:phi-1}.

For extended-valued costs, apply the finite-cost result to $c^{(M)}_{i}:=\min\{c_{i},M\}$
and let $M\to\infty$ or equivalently take supremum over $M>0$. By
swapping the supremum over $M>0$ and the supremum in optimal transport,
we can find that the optimal transport cost induced by $c^{(M)}_{i}$
increases to the optimal transport cost induced by $c_{i}$ as $M\to\infty$.
It means \eqref{eq:transport-1} still holds and thus, so does \eqref{eq:phi-1}.

\end{IEEEproof}

\subsection{Additivity of $\Gamma$}

We first prove the superadditivity of $\Gamma$. Let 
\begin{equation}
P_{WX_{1}X_{2}Y_{1}Y_{2}}=P_{W}P_{X_{1}X_{2}|W}P_{Y_{1}Y_{2}|W}\label{eq:two-letter}
\end{equation}
be an arbitrary feasible synthesis distribution satisfying
\begin{equation}
P_{X_{1}X_{2}Y_{1}Y_{2}}(x_{1},x_{2},y_{1},y_{2})=\pi_{1}(x_{1},y_{1})\pi_{2}(x_{2},y_{2}).\label{eq:product}
\end{equation}
Set $c_{i}=c_{\pi_{i}}$, and thus, $c_{\pi_{1}\otimes\pi_{2}}=c_{1}+c_{2}$.
Applying Lemma~\ref{lem:superadditivity} separately at each $w$
and taking average with respect to $P_{W}$ gives 
\begin{align}
 & \E_{W}\Phi_{c_{1}+c_{2}}\bigl(P_{X_{1}X_{2}|W},P_{Y_{1}Y_{2}|W}\bigr)\nonumber \\
\geq & \E_{W,Y_{2}}\Phi_{c_{1}}\bigl(P_{X_{1}|W},P_{Y_{1}|W,Y_{2}}\bigr)+\E_{W,X_{1}}\Phi_{c_{2}}\bigl(P_{X_{2}|W,X_{1}},P_{Y_{2}|W}\bigr).\label{eq:gluing}
\end{align}

Define $U=(W,Y_{2}),\,V=(W,X_{1}).$ By conditional independence in
\eqref{eq:two-letter}, 
\begin{align}
P_{X_{1}Y_{1}|U} & =P_{X_{1}|W}P_{Y_{1}|W,Y_{2}},\label{eq:markov1}\\
P_{X_{2}Y_{2}|V} & =P_{X_{2}|W,X_{1}}P_{Y_{2}|W}.\label{eq:markov2}
\end{align}
That is, $X_{1}\leftrightarrow U\leftrightarrow Y_{1}$ and $X_{2}\leftrightarrow V\leftrightarrow Y_{2}$.
Their respective source marginals are $\pi_{1}$ and $\pi_{2}$ by
\eqref{eq:product}. Hence both are feasible single-letter synthesis
distributions, and thus,
\begin{align}
\E_{W,Y_{2}}\Phi_{c_{1}}\bigl(P_{X_{1}|W},P_{Y_{1}|W,Y_{2}}\bigr) & \geq\Gamma(\pi_{1}),\label{eq:gamma1}\\
\E_{W,X_{1}}\Phi_{c_{2}}\bigl(P_{X_{2}|W,X_{1}},P_{Y_{2}|W}\bigr) & \geq\Gamma(\pi_{2}).\label{eq:gamma2}
\end{align}
Combining \eqref{eq:gluing}--\eqref{eq:gamma2} and taking the infimum
over arbitrary feasible distributions yields 
\begin{equation}
\Gamma(\pi_{1}\otimes\pi_{2})\geq\Gamma(\pi_{1})+\Gamma(\pi_{2}).\label{eq:superadditivity}
\end{equation}

For the reverse direction, subadditivity is easy to verify. Observe
that 
\begin{equation}
\cT_{c_{1}+c_{2}}(P_{1}\otimes P_{2},Q_{1}\otimes Q_{2})=\cT_{c_{1}}(P_{1},Q_{1})+\cT_{c_{2}}(P_{2},Q_{2}).\label{eq:product-transport}
\end{equation}
Since entropy is additive under product distributions, we obtain 
\begin{equation}
\Phi_{c_{1}+c_{2}}(P_{1}\otimes P_{2},Q_{1}\otimes Q_{2})=\Phi_{c_{1}}(P_{1},Q_{1})+\Phi_{c_{2}}(P_{2},Q_{2}).\label{eq:product-phi}
\end{equation}
Taking independent feasible syntheses for $\pi_{1}$ and $\pi_{2}$
gives 
\begin{equation}
\Gamma(\pi_{1}\otimes\pi_{2})\leq\Gamma(\pi_{1})+\Gamma(\pi_{2}).\label{eq:subadditivity}
\end{equation}

\section{Proof of Theorem \ref{thm:ECS} }

We first prove $\mathcal{R}_{\mathrm{S}}(\pi_{1}\otimes\pi_{2})\subseteq\mathcal{R}_{\mathrm{S}}(\pi_{1})+\mathcal{R}_{\mathrm{S}}(\pi_{2}).$
Similarly to the common information case, let 
\begin{equation}
P_{WX_{1}X_{2}Y_{1}Y_{2}}=P_{W}P_{X_{1}X_{2}|W}P_{Y_{1}Y_{2}|W}\label{eq:two-letter-1}
\end{equation}
be an arbitrary feasible synthesis distribution satisfying
\begin{equation}
P_{X_{1}X_{2}Y_{1}Y_{2}}(x_{1},x_{2},y_{1},y_{2})=\pi_{1}(x_{1},y_{1})\pi_{2}(x_{2},y_{2}).\label{eq:product-1}
\end{equation}
Set $c_{i}=c_{\pi_{i}}$, and thus, $c_{\pi_{1}\otimes\pi_{2}}=c_{1}+c_{2}$.
Then, \eqref{eq:gluing} still holds. Substituting $U=(W,Y_{2}),\,V=(W,X_{1})$,
we obtain 
\begin{align}
 & \E_{W}\Phi_{c_{1}+c_{2}}\bigl(P_{X_{1}X_{2}|W},P_{Y_{1}Y_{2}|W}\bigr)\nonumber \\
\geq & \E_{U}\Phi_{c_{1}}\bigl(P_{X_{1}|U},P_{Y_{1}|U}\bigr)+\E_{V}\Phi_{c_{2}}\bigl(P_{X_{2}|V},P_{Y_{2}|V}\bigr).\label{eq:gluing-1}
\end{align}

In addition, 
\begin{align}
 & I(W;X_{1}X_{2})-\left(I(U;X_{1})+I(V;X_{2})\right)\nonumber \\
 & =H(X_{1}X_{2})+H(W)-H(WX_{1}X_{2})\nonumber \\
 & \qquad-H(WY_{2})-H(X_{1})+H(WX_{1}Y_{2})\nonumber \\
 & \qquad-H(WX_{1})-H(X_{2})+H(WX_{1}X_{2})\nonumber \\
 & =0,\label{eq:R}
\end{align}
where the last line follows since 1) $P_{X_{1}X_{2}Y_{1}Y_{2}}=\pi_{1}\otimes\pi_{2}$
and thus, 
\[
H(X_{1}X_{2})=H(X_{1})+H(X_{2});
\]
2) $X_{1}\leftrightarrow W\leftrightarrow Y_{2}$ and thus, 
\begin{align*}
H(W)+H(WX_{1}Y_{2}) & =H(W)+H(WX_{1})+H(Y_{2}|W)\\
 & =H(WX_{1})+H(WY_{2}).
\end{align*}

Combining \eqref{eq:gluing-1} and \eqref{eq:R} yields $\mathcal{R}_{\mathrm{S}}(\pi_{1}\otimes\pi_{2})\subseteq\mathcal{R}_{\mathrm{S}}(\pi_{1})+\mathcal{R}_{\mathrm{S}}(\pi_{2}).$ 

The reverse direction is easy to verify. Consider independent feasible
syntheses for $\pi_{1}$ and $\pi_{2}$: 
\begin{equation}
P_{UVX_{1}X_{2}Y_{1}Y_{2}}=P_{U}P_{X_{1}Y_{1}|U}P_{V}P_{X_{2}Y_{2}|V}.\label{eq:two-letter-1-1}
\end{equation}
Observe that 
\begin{equation}
\Phi_{c_{1}+c_{2}}(P_{X_{1}|U}\otimes P_{X_{2}|V},P_{Y_{1}|U}\otimes P_{Y_{2}|V})=\Phi_{c_{1}}(P_{X_{1}|U},P_{Y_{1}|U})+\Phi_{c_{2}}(P_{X_{2}|V},P_{Y_{2}|V}).\label{eq:product-phi-1}
\end{equation}
Moreover, $I(UV;X_{1}X_{2})=I(U;X_{1})+I(V;X_{2})$. Since $P_{UVX_{1}X_{2}Y_{1}Y_{2}}$
is a feasible synthesis for $\pi_{1}\otimes\pi_{2}$, we obtain 
\begin{equation}
\mathcal{R}_{\mathrm{S}}(\pi_{1}\otimes\pi_{2})\supseteq\mathcal{R}_{\mathrm{S}}(\pi_{1})+\mathcal{R}_{\mathrm{S}}(\pi_{2}).\label{eq:subadditivity-1}
\end{equation}

\subsection*{Acknowledgements }

\emph{Generative-AI use disclosure:} ChatGPT 6 was used to construct
opposite-order couplings and refine exposition of the paper. All mathematical
claims are the full responsibility of the author. 

\emph{Funding:} This work was supported by the National Key Research
and Development Program of China under grant 2023YFA1009604 and the
NSFC under grant 62101286.

\bibliographystyle{unsrt}
\bibliography{ref}

\begin{thebibliography}{10}

\bibitem{WynerCI}
A.~D. Wyner.
\newblock The common information of two dependent random variables.
\newblock {\em IEEE Transactions on Information Theory}, 21(2):163--179, Mar
  1975.

\bibitem{KLE2014}
G.~R. Kumar, C.-T. Li, and A.~{El Gamal}.
\newblock Exact common information.
\newblock In {\em IEEE International Symposium on Information Theory (ISIT)},
  pages 161--165, Honolulu, Hawaii, USA, 2014.

\bibitem{Vellambi2018}
B.~N. Vellambi and J.~Kliewer.
\newblock New results on the equality of exact and {Wyner} common information
  rates.
\newblock In {\em IEEE International Symposium on Information Theory (ISIT)},
  pages 151--155, Vail, Colorado, USA, Oct 2018.

\bibitem{YuTan2020_exact}
L.~Yu and V.~Y.~F. Tan.
\newblock On exact and {$\infty$}-{R\'enyi} common information.
\newblock {\em IEEE Transactions on Information Theory}, 66(6):3366--3406, Jun
  2020.

\bibitem{LiElgamal2017}
C.-T. Li and A.~{El Gamal}.
\newblock Distributed simulation of continuous random variables.
\newblock {\em IEEE Transactions on Information Theory}, 63(10):6329--6343, Oct
  2017.

\bibitem{YuTan2018}
L.~Yu and V.~Y.~F. Tan.
\newblock Wyner's common information under {R}\'enyi divergence measures.
\newblock {\em IEEE Transactions on Information Theory}, 64(5):3616--3623, May
  2018.

\bibitem{yu2020corrections}
L.~Yu and V.~Y.~F. Tan.
\newblock {Corrections to ``Wyner's common information under R{\'e}nyi
  divergence measures''}.
\newblock {\em IEEE Transactions on Information Theory}, 66(4):2599--2608,
  2020.

\bibitem{yu2026exact}
Lei Yu.
\newblock Exact common information and exact channel synthesis for correlated
  gaussian sources, 2026.

\bibitem{Bennett02}
C.~H. Bennett, P.~W. Shor, J.~A. Smolin, and A.~V. Thapliyal.
\newblock Entanglement-assisted capacity of a quantum channel and the reverse
  {Shannon} theorem.
\newblock {\em IEEE Transactions on Information Theory}, 48(10):2637--2655, Oct
  2002.

\bibitem{Win02}
A.~Winter.
\newblock Compression of sources of probability distributions and density
  operators.
\newblock {\em arXiv:quant-ph/0208131}, 2002.

\bibitem{cuff13}
P.~Cuff.
\newblock Distributed channel synthesis.
\newblock {\em IEEE Transactions on Information Theory}, 59(11):7071--7096,
  2013.

\bibitem{bennett14quantum}
C.~H. Bennett, I.~Devetak, A.~W. Harrow, P.~W. Shor, and A.~Winter.
\newblock The quantum reverse {Shannon} theorem and resource tradeoffs for
  simulating quantum channels.
\newblock {\em IEEE Transactions on Information Theory}, 60(3):2926--2959, Oct
  2014.

\bibitem{Harsha10}
P.~Harsha, R.~Jain, D.~McAllester, and J.~Radhakrishnan.
\newblock The communication complexity of correlation.
\newblock {\em IEEE Transactions on Information Theory}, 56(1):438--449, Oct
  2010.

\bibitem{YuTan2020b}
L.~Yu and V.~Y.~F. Tan.
\newblock Exact channel synthesis.
\newblock {\em IEEE Transactions on Information Theory}, 66(5):2299--2818, May
  2020.

\bibitem{LiElgamal2018}
C.-T. Li and A.~{El Gamal}.
\newblock Strong functional representation lemma and applications to coding
  theorems.
\newblock {\em IEEE Transactions on Information Theory}, 64(11):6967--6978, Oct
  2018.

\bibitem{sriramu2024optimal}
S.~M. Sriramu and A.~B. Wagner.
\newblock Optimal redundancy in exact channel synthesis.
\newblock In {\em 2024 IEEE International Symposium on Information Theory
  (ISIT)}, pages 2401--2406. IEEE, 2024.

\bibitem{Cubitt11}
T.~S. Cubitt, D.~Leung, W.~Matthews, and A.~Winter.
\newblock Zero-error channel capacity and simulation assisted by non-local
  correlations.
\newblock {\em IEEE Transactions on Information Theory}, 57(8):5509--5523, Mar
  2011.

\bibitem{cao2024channelsimulationfinite}
M.~X. Cao, N.~Ramakrishnan, M.~Berta, and M.~Tomamichel.
\newblock Channel simulation: Finite blocklengths and broadcast channels.
\newblock {\em IEEE Transactions on Information Theory}, 70(10):6780--6808,
  2024.

\bibitem{li2026largedevrst}
S.‑B. Li, K.~Li, and L.~Yu.
\newblock Large deviation analysis for the reverse shannon theorem.
\newblock {\em IEEE Transactions on Information Theory}, 72(8):5335--5355, aug
  2026.

\bibitem{oufkir2026exponentsrandomness}
A.~Oufkir, M.~X. Cao, H.-C. Cheng, and M.~Berta.
\newblock Exponents for shared randomness-assisted channel simulation.
\newblock {\em IEEE Transactions on Information Theory}, 72(5):2624--2640,
  2026.

\bibitem{Oufkir2026}
A.~Oufkir, Y.~Yao, and M.~Berta.
\newblock Exponents for classical-quantum channel simulation in purified
  distance.
\newblock {\em IEEE Transactions on Information Theory}, 72(7):4537--4554,
  2026.

\bibitem{yu2022common}
L.~Yu and V.~Y.~F. Tan.
\newblock Common information, noise stability, and their extensions.
\newblock {\em Foundations and Trends in Communications and Information
  Theory}, 19(2):107--389, 2022.

\bibitem{li2024channelsimulation}
C.~T. Li.
\newblock Channel simulation: Theory and applications to lossy compression and
  differential privacy.
\newblock {\em Foundations and Trends in Communications and Information
  Theory}, 21(6):847--1106, 2024.

\end{thebibliography}

\end{document}